\documentclass[pdflatex,sn-mathphys-num]{sn-jnl}

\usepackage[T1]{fontenc}
\usepackage{lmodern}
\usepackage{microtype}
\usepackage{amsmath,amssymb,amsthm,mathtools}
\usepackage{booktabs}
\usepackage{algorithm}
\usepackage[noend]{algpseudocode}
\usepackage{enumitem}
\usepackage{xcolor}
\usepackage{tikz}
\usetikzlibrary{arrows.meta,positioning,fit}
\usepackage[nameinlink,capitalize,noabbrev]{cleveref}

\makeatletter
\renewcommand{\theHALG@line}{\thealgorithm.\arabic{ALG@line}}
\makeatother

\hypersetup{
  colorlinks=true,
  linkcolor=blue!55!black,
  citecolor=green!40!black,
  urlcolor=blue!60!black,
  pdfauthor={Hao Lu, Yuan Yuan, Xingwu Liu, Xin Han},
  pdftitle={Two-Machine Flow Shop with a Fixed Non-Availability Interval on the Second Machine}
}

\newtheorem{theorem}{Theorem}[section]
\newtheorem{lemma}[theorem]{Lemma}

\theoremstyle{definition}

\newtheorem{remark}[theorem]{Remark}

\newcommand{\J}{\mathcal{J}}
\newcommand{\E}{\mathcal{E}}
\newcommand{\V}{\mathcal{V}}
\newcommand{\OPT}{\mathrm{OPT}}
\newcommand{\CJ}{C^{\mathrm J}}

\newcommand{\ind}{\mathbf{1}}
\newcommand{\disj}{\mathbin{\dot\cup}}

\begin{document}

\title[Two-Machine Flow Shop with a Fixed Non-Availability Interval]{Two-Machine
Flow Shop with a Fixed Non-Availability Interval on the Second Machine}

\author[1]{\fnm{Hao} \sur{Lu}}
\author*[2]{\fnm{Yuan} \sur{Yuan}\,\href{https://orcid.org/0000-0001-5403-6576}{\textsuperscript{ORCID}}}
\email{yuanyuan12430@126.com}
\author[3]{\fnm{Xingwu} \sur{Liu}}
\author[1]{\fnm{Xin} \sur{Han}\,\href{https://orcid.org/0000-0002-1694-7712}{\textsuperscript{ORCID}}}

\affil*[1]{\orgdiv{School of Software},
  \orgname{Dalian University of Technology},
  \orgaddress{\city{Dalian}, \postcode{116620}, \state{Liaoning},
  \country{China}}}

\affil[2]{\orgdiv{School of Information and Communication Engineering},
  \orgname{Dalian Minzu University},
  \orgaddress{\city{Dalian}, \postcode{116600}, \state{Liaoning},
  \country{China}}}

\affil[3]{\orgdiv{School of Mathematical Sciences},
  \orgname{Dalian University of Technology},
  \orgaddress{\city{Dalian}, \postcode{116024}, \state{Liaoning},
  \country{China}}}

\abstract{
This paper investigates a two-machine permutation flow shop in which the second
machine is unavailable during one fixed interval $[s,t]$.  We consider the
non-resumable setting: an operation interrupted by the interval must restart
from the beginning after the machine becomes available.  The objective is to
minimize the makespan.  We establish three results.  First, we give a
polynomial-time $10/7$-approximation algorithm.  Second, we develop a
pseudopolynomial-time exact dynamic program.  Third, we prove that the problem
does not admit a fully polynomial-time approximation scheme (FPTAS) unless
$\mathrm{P}=\mathrm{NP}$, even when the non-availability interval has unit
length.  Together, these results characterize a distinctive complexity profile:
exact optimization is possible in pseudopolynomial time, whereas the usual
route from such an algorithm to an FPTAS is impossible unless
$\mathrm{P}=\mathrm{NP}$.  They also reveal an approximability separation from
the corresponding non-resumable problem with the interval on the first machine.
}


\maketitle

\section{Introduction}

Machine non-availability is unavoidable in many production systems.  Preventive
maintenance, calibration, cleaning, tool replacement, and the reserved use of a
machine may all interrupt an otherwise continuous processing horizon.  When the
interruption is known in advance, it is naturally modeled as a non-availability
interval.  Such an interval is \emph{fixed} when both endpoints are prescribed;
in a flexible model, its starting time is chosen within a given window while its
duration is fixed.  We consider a single fixed interval and refer to it
interchangeably as a maintenance period.

The classical two-machine flow shop consists of jobs that must first be
processed on machine $M_1$ and then on machine $M_2$.  Job $i$ requires $a_i$
time units on $M_1$ and $b_i$ time units on $M_2$, and the objective is to
minimize the makespan, that is, the completion time of the last job.  If both
machines are continuously available, Johnson's rule constructs an optimal
permutation in $O(n\log n)$ time~\cite{Johnson1954}: jobs with $a_i\le b_i$ are
placed first in nondecreasing order of $a_i$, and the remaining jobs are placed
last in nonincreasing order of $b_i$.  A prescribed interruption destroys this
direct solution because the schedule must determine not only an order but also
which jobs use the capacity before the interval.

The treatment of an operation that meets the interval leads to three standard
models.  In the resumable model, the operation continues after the machine
becomes available.  In the semi-resumable model, only part of the work performed
before the interruption is lost.  In the non-resumable model studied here, an
interrupted operation must restart from the beginning.  We use the notation
\[
  F2\mid \mathrm{nr\mbox{-}a}(M_2)\mid C_{\max}.
\]
The second machine is unavailable during $[s,t]$, where $0\le s<t$, while the
first machine remains continuously available.  All processing times are
positive integers, all jobs are available at time zero, and we consider
permutation schedules.  Because processing performed on an interrupted
$M_2$-operation can be discarded without delaying any completion, the essential
decision is a partition: one set of jobs completes on $M_2$ no later than $s$,
and the remaining jobs start on $M_2$ no earlier than $t$.  Johnson's rule still
orders each fixed side optimally, but it does not determine the partition.

Machine scheduling with availability constraints was studied systematically by
Lee~\cite{Lee1996}.  For two-machine flow shops, the foundational papers of Lee
established hardness results, pseudopolynomial algorithms, and constant-factor
approximation guarantees for several interruption models~\cite{Lee1997,Lee1999}.
Kubiak et al. investigated structural and algorithmic aspects of two-machine
flow shops with limited machine availability~\cite{KubiakEtAl2002}, and Mosheiov
et al. later considered flow-shop and open-shop models with a single maintenance
window~\cite{MosheiovEtAl2018}.

For the resumable two-machine flow shop with one fixed interval, Lee obtained
approximation algorithms for intervals on either machine~\cite{Lee1997}.  Breit
improved the guarantee for an interval on $M_2$~\cite{Breit2004}, whereas Cheng
and Wang improved the guarantee for an interval on $M_1$~\cite{ChengWang2000}.
Ng and Kovalyov subsequently gave fully polynomial-time approximation schemes
for both machine locations~\cite{NgKovalyov2004}.  Further polynomial-time
approximation schemes have been developed for related flow-shop models with one
or several availability constraints~\cite{Breit2006,Hadda2012}.  Lee and Kubzin,
Potts, and Strusevich also obtained constant-factor algorithms and approximation
schemes for semi-resumable and more general availability
models~\cite{Lee1999,KubzinPottsStrusevich2009}.

Availability constraints have also been studied extensively in the closely
related two-machine open shop.  Breit, Schmidt, and Strusevich analyzed both the
availability-constraint model and its non-preemptive counterpart
~\cite{BreitSchmidtStrusevich2001,BreitSchmidtStrusevich2003}.  Lorigeon,
Billaut, and Bouquard developed a pseudopolynomial dynamic program
~\cite{LorigeonEtAl2002}; Kubzin et al. established polynomial-time approximation
schemes~\cite{KubzinStrusevichEtAl2006}; and Yuan et al. gave a PTAS for the
non-resumable open-shop variant~\cite{YuanEtAl2022}.  These results are relevant
comparators, although the fixed machine order in a flow shop leads to different
critical-path and partition structures.

The non-resumable setting is more sensitive to the location of the interval.
For an interval on $M_1$, Hadda, Dridi, and Hajri-Gabouj obtained a
$3/2$-approximation~\cite{HaddaDridiHajri2010}.  More recently, Lu, Han, Yuan,
Liu, and Yang improved this ratio to $7/5$ and developed the first fully
polynomial-time approximation scheme for that model~\cite{LuEtAl2026}.  The
corresponding problem with the interval on $M_2$ has a different structure.  A
choice of jobs that finish before $s$ changes both the amount of second-machine
work postponed beyond $t$ and the first-machine offset of every critical path in
the late Johnson subsequence.  These two effects need not move in the same
direction, and exchanging the roles of $a_i$ and $b_i$ does not preserve the
interaction with the fixed interval.

For completeness, a polynomial-time algorithm is a $\rho$-approximation if it
always returns a schedule of makespan at most $\rho$ times the optimum.  A
polynomial-time approximation scheme (PTAS) returns a $(1+\varepsilon)$-
approximate schedule for every fixed $\varepsilon>0$.  It is a fully
polynomial-time approximation scheme (FPTAS) if its running time is polynomial
in both the input length and $1/\varepsilon$.  In many weakly NP-hard problems,
a pseudopolynomial exact algorithm can be converted into an FPTAS by scaling or
trimming.  One of our main conclusions is that this familiar implication fails
for $F2\mid \mathrm{nr\mbox{-}a}(M_2)\mid C_{\max}$.

\noindent\textbf{Our contributions.}
We establish the following three results.

\begin{enumerate}[leftmargin=2.2em,label=\textbf{(C\arabic*)}]
  \item We give a polynomial-time $10/7$-approximation algorithm.  A direct
  implementation runs in $O(n^2\log n)$ time and uses $O(n)$ auxiliary space.

  \item We develop an exact dynamic program with running time $O(nABP^2)$ and
  memory requirement $O(ABP^2)$, where $A=\sum_i a_i$, $B=\sum_i b_i$, and
  $P=A+B$.  Hence the problem is exactly solvable in pseudopolynomial time.

  \item We prove that the problem admits no FPTAS unless
  $\mathrm{P}=\mathrm{NP}$.  The result holds even when the interval has unit
  length, $t-s=1$.
\end{enumerate}

The second and third results together identify an unusual approximability
boundary: pseudopolynomial exact solvability coexists with the nonexistence of an
FPTAS.  Combined with the recent FPTAS for an interruption on $M_1$, they also
show that the location of a fixed non-availability interval fundamentally
changes the approximability of the non-resumable two-machine flow shop.

\noindent\textbf{Organization.}
\Cref{sec:preliminaries} introduces the structural decomposition used throughout
the paper.  \Cref{sec:approximation} presents the $10/7$-approximation algorithm,
\cref{sec:dp} gives the pseudopolynomial exact dynamic program, and
\cref{sec:no-fptas} proves the FPTAS lower bound.  We conclude in
\cref{sec:conclusion}.

\section{Preliminaries and structural decomposition}
\label{sec:preliminaries}

Let $\J=\{1,\ldots,n\}$.  Job $i$ requires $a_i>0$ time units on $M_1$ and
$b_i>0$ time units on $M_2$.  For $Q\subseteq\J$, write
\[
  a(Q)=\sum_{i\in Q}a_i,
  \qquad
  b(Q)=\sum_{i\in Q}b_i.
\]

Partition the jobs into the Johnson-early and Johnson-late classes
\[
  \E=\{i:a_i\le b_i\},
  \qquad
  \V=\{i:a_i>b_i\}.
\]
Fix once and for all a total Johnson order $\prec_{\mathrm J}$: jobs in $\E$ are
sorted by nondecreasing $a_i$, jobs in $\V$ are sorted by nonincreasing $b_i$,
and every early job precedes every late job.  Ties are resolved consistently.
When equal-$b$ late jobs include both a job designated large and one designated
small in \cref{sec:algorithm}, the large job comes first.  For every
$Q\subseteq\J$, let $J(Q)$ denote the restriction of this global order to $Q$,
and let $\CJ(Q)$ be the classical two-machine makespan of $J(Q)$ without the
non-availability interval.  By Johnson's theorem, $\CJ(Q)$ is minimum among all
orders of $Q$.

For $q\in Q$, define the Johnson cut
\[
  \Psi_q(Q)=a(Q_{\preceq q})+b(Q_{\succeq q}),
\]
where the two subsets are taken in the restricted global Johnson order.  The
standard critical-path representation gives
\begin{equation}
  \CJ(Q)=\max_{q\in Q}\Psi_q(Q).
  \label{eq:johnson-cut}
\end{equation}
A maximizer is called a \emph{critical job}.  Relative to a fixed anchor $q$, a
job $i\ne q$ contributes $a_i$ if $i\prec_{\mathrm J}q$ and $b_i$ if
$q\prec_{\mathrm J}i$; the anchor contributes $a_q+b_q$.

An early/late partition is denoted $\J=X\disj Y$: jobs in $X$ complete their
second operations no later than $s$, and jobs in $Y$ are processed on $M_2$ no
earlier than $t$.  The corresponding canonical schedule is
\[
  J(X)\mid[s,t]\mid J(Y).
\]

\begin{lemma}
\label{lem:canonical}
There is an optimal schedule induced by a partition
$\J=X^*\disj Y^*$ such that
\begin{equation}
  \CJ(X^*)\le s,
  \label{eq:early-feasible}
\end{equation}
and both $X^*$ and $Y^*$ occur in their restricted global Johnson orders.
\end{lemma}

\begin{proof}
Consider an optimal schedule.  If a second operation begins before $s$ but is
interrupted and restarted after $t$, deleting its processing before $s$ does not
delay any completion.  Hence we may assume that every second operation either
finishes by $s$ or starts at or after $t$.  Let $X^*$ be the jobs of the first
kind and let $Y^*$ contain the remaining jobs.  The jobs of $X^*$ precede those
of $Y^*$ on both machines in the induced permutation schedule.

Replace the order within $X^*$ by $J(X^*)$.  Johnson optimality cannot increase
the completion time of the early block, so it remains at most $s$.  Once the
total first-machine offset $a(X^*)$ and the second-machine release time $t$ are
fixed, replacing the order within $Y^*$ by $J(Y^*)$ minimizes the only
order-dependent two-machine term for the late block.  The replacement therefore
does not increase the makespan.  The resulting schedule has the asserted form.
\end{proof}

\begin{lemma}
\label{lem:partition-makespan}
If $\J=X\disj Y$ and $\CJ(X)\le s$, then the canonical schedule has makespan
\begin{equation}
 C(X,Y)=
 \begin{cases}
   \CJ(X), & Y=\varnothing,\\[2mm]
   \max\{t+b(Y),\ a(X)+\CJ(Y)\}, & Y\ne\varnothing.
 \end{cases}
 \label{eq:partition-makespan}
\end{equation}
\end{lemma}

\begin{proof}
The empty-$Y$ case is immediate.  Otherwise, expand the usual two-machine
recurrence along $J(Y)=(i_1,\ldots,i_\ell)$.  Machine $M_1$ enters this block
with offset $a(X)$, whereas $M_2$ cannot start it before $t$.  Thus
\[
 C(X,Y)=\max\left\{t+b(Y),\
 a(X)+\max_{1\le k\le\ell}
 \left(\sum_{h\le k}a_{i_h}+\sum_{h\ge k}b_{i_h}\right)\right\},
\]
and the inner maximum equals $\CJ(Y)$ by \cref{eq:johnson-cut}.
\end{proof}

When $q\in Y$, it will be convenient to use the anchored expression
\begin{equation}
 F_q(X,Y)=a(X)+a(Y_{\preceq q})+b(Y_{\succeq q}).
 \label{eq:anchored-cut}
\end{equation}
If $q$ is critical in $J(Y)$, then
$F_q(X,Y)=a(X)+\CJ(Y)$.  For another partition containing $q$ on its late side,
$F_q$ remains a valid cut value even if $q$ is not critical there.

\section{A \texorpdfstring{$10/7$}{10/7}-approximation algorithm}
\label{sec:approximation}

\subsection{The algorithm}
\label{sec:algorithm}

Let $L$ contain the six jobs with largest $b_i$ values, or all jobs if $n\le6$,
and let $S=\J\setminus L$.  Call these jobs \emph{large} and \emph{small},
respectively.  Sort the small jobs by the ratio rule
\[
  \rho_i=\frac{b_i}{a_i},
  \qquad
  j_1\prec_{\mathrm R}j_2\prec_{\mathrm R}\cdots\prec_{\mathrm R}j_m,
\]
in nonincreasing order, with ties consistent with the global conventions.
The ratio order is used only to choose membership in a partition.  Every
candidate is actually sequenced by the restricted Johnson order.

Algorithm $H$ enumerates all $2^{|L|}\le64$ partitions
$L=L_X\disj L_Y$.  For one such enumeration, if $\CJ(L_X)>s$, it is discarded.
Otherwise, choose the largest $r\in\{0,\ldots,m\}$ such that
\begin{equation}
 \CJ\bigl(L_X\cup\{j_1,\ldots,j_r\}\bigr)\le s.
 \label{eq:max-prefix}
\end{equation}

If $r<m$, set
\[
 X=L_X\cup\{j_1,\ldots,j_r\},
 \quad
 Y=L_Y\cup\{j_{r+1},\ldots,j_m\},
 \quad
 p=j_{r+1}.
\]
The job $p$ is the \emph{boundary job}.  The algorithm constructs the basic
candidate $\pi^0=(J(X),J(Y))$ and the purge candidate
\[
 X^-=X\setminus\{i\in X\cap S:a_i>b_i\},
 \qquad
 Y^-=\J\setminus X^-.
\]

If $r=m$, define
\[
 S^{\mathrm E}=\{i\in S:a_i\le b_i\},
 \quad
 S^{\mathrm L}=\{i\in S:a_i>b_i\},
\]
and split the small late jobs at ratio $1/2$:
\[
 D^+=\left\{i\in S^{\mathrm L}:\frac{b_i}{a_i}\ge\frac12\right\},
 \qquad
 D^-=S^{\mathrm L}\setminus D^+.
\]
It constructs three candidates:
\begin{align}
 (X_1,Y_1)&=(L_X\cup S,\ L_Y), \label{eq:no-boundary-c1}\\
 (X_2,Y_2)&=(L_X\cup S^{\mathrm E},\ L_Y\cup S^{\mathrm L}),
 \label{eq:no-boundary-c2}\\
 (X_3,Y_3)&=(L_X\cup S^{\mathrm E}\cup D^+,\ L_Y\cup D^-).
 \label{eq:no-boundary-c3}
\end{align}
All sets in \cref{eq:no-boundary-c1,eq:no-boundary-c2,eq:no-boundary-c3}
are sequenced by Johnson's rule.  Finally, $H$ returns the feasible candidate
with minimum makespan.

\begin{algorithm}[t]
\caption{$H$}
\label{alg:h}
\begin{algorithmic}[1]
\State Let $L$ be the six jobs with largest $b_i$ and let $S=\J\setminus L$.
\State Sort $S$ as $j_1,\ldots,j_m$ by nonincreasing $b_i/a_i$.
\For{each partition $L=L_X\disj L_Y$}
  \If{$\CJ(L_X)>s$}
    \State \textbf{continue}
  \EndIf
  \State Find the maximum $r$ satisfying \cref{eq:max-prefix}.
  \If{$r<m$}
    \State Add the basic candidate $(X,Y)$ and the purge candidate $(X^-,Y^-)$.
  \Else
    \State Add the three candidates in
    \cref{eq:no-boundary-c1,eq:no-boundary-c2,eq:no-boundary-c3}.
  \EndIf
\EndFor
\State \Return a generated candidate of minimum makespan.
\end{algorithmic}
\end{algorithm}

Fix the optimal canonical partition $(X^*,Y^*)$ from
\cref{lem:canonical}.  Since every placement of the large jobs is enumerated,
one iteration satisfies
\begin{equation}
  L_X=L\cap X^*,
  \qquad
  L_Y=L\cap Y^*.
  \label{eq:correct-enumeration}
\end{equation}
We call it the \emph{correct enumeration} and analyze only its candidates.

\subsection{The small-job scale and tail discrepancy}

\begin{lemma}
\label{lem:small-bound}
For every $i\in S$,
\[
  b_i\le \frac17\OPT.
\]
\end{lemma}

\begin{proof}
There are six distinct large jobs whose $b$ values are at least $b_i$.
Therefore $b(\J)\ge7b_i$.  Every feasible schedule must complete all second
operations, so $\OPT\ge b(\J)\ge7b_i$.
\end{proof}

Set
\begin{equation}
  \beta=\frac17\OPT.
  \label{eq:beta}
\end{equation}
Thus every small job has $b_i\le\beta$.

Suppose the correct enumeration has a boundary job $p$.  Put
$Q=X\cup\{p\}$.  Maximality of the prefix gives
\begin{equation}
  \CJ(X)\le s<\CJ(Q).
  \label{eq:boundary-infeasible}
\end{equation}
Define the excess late load
\begin{equation}
 E_T=b(Y)-b(Y^*)=b(X^*)-b(X).
 \label{eq:tail-error}
\end{equation}
If
\[
 U=(X^*\setminus X)\cap S,
 \qquad
 W=(X\setminus X^*)\cap S,
\]
then the large parts cancel under the correct enumeration and
\begin{equation}
  E_T=b(U)-b(W).
  \label{eq:tail-error-uw}
\end{equation}

\begin{lemma}
\label{lem:tail-discrepancy}
If the boundary job exists, then $E_T\le2\beta$.
\end{lemma}

\begin{proof}
Let $c$ be a critical job of $J(Q)$, so
$\Psi_c(Q)=\CJ(Q)>s$.  We distinguish five exhaustive cases.

\smallskip
\noindent\emph{Case 1: $c$ is early and $a_c\le\beta$.}
Every job before $c$ is early.  Hence
\[
 \CJ(Q)=\Psi_c(Q)\le b(Q)+a_c=b(X)+b_p+a_c.
\]
Since $b(X^*)\le\CJ(X^*)\le s$,
\[
 E_T=b(X^*)-b(X)\le s-b(X)<b_p+a_c\le2\beta.
\]

\smallskip
\noindent\emph{Case 2: $c$ is early and $a_c>\beta$.}
The job $c$ is large, because a small early job satisfies
$a_c\le b_c\le\beta$.  Thus $c\in L_X\subseteq X^*$.  Every small early job
precedes $c$, and every small late job follows $c$.  Consequently each small job
$i$ contributes
\[
  w_i=\min\{a_i,b_i\}
\]
to the fixed $c$-cut.  Give item $i$ profit $b_i$.  Its density is $b_i/a_i$
for an early job and $1$ for a late job; among density-one jobs use the ratio
order as tie-breaking.  After subtracting the fixed large-job contribution from
$s$, \cref{eq:boundary-infeasible} says that $X\cap S$ is the complete greedy
prefix and $p$ is the first item that does not fit.  Feasibility of $X^*$ makes
$X^*\cap S$ an integral packing for the same residual capacity.  The fractional
knapsack bound therefore yields
\[
 b(X^*\cap S)\le b(X\cap S)+b_p,
\]
so $E_T\le b_p\le\beta$.

\smallskip
\noindent\emph{Case 3: $c$ is late, $\rho_p<1$, and $b_c\le\beta$.}
The ratio-prefix property implies
\[
 u\in U\Rightarrow\rho_u\le\rho_p,
 \qquad
 w\in W\Rightarrow\rho_w\ge\rho_p.
\]
Using cancellation of the large parts and $a(X^*)\le s$,
\begin{align}
 E_T
 &=b(U)-b(W)\\
 &\le\rho_p\bigl(a(U)-a(W)\bigr)
 =\rho_p\bigl(a(X^*)-a(X)\bigr)
 \le\rho_p(s-a(X)).
 \label{eq:case3-ratio}
\end{align}
Since $c$ is late, every job after $c$ is late, and
\[
 \CJ(Q)=\Psi_c(Q)\le a(Q)+b_c=a(X)+a_p+b_c.
\]
Together with \cref{eq:boundary-infeasible}, this gives
$s-a(X)<a_p+b_c$.  Substitution in \cref{eq:case3-ratio} yields
\[
 E_T<\rho_p(a_p+b_c)=b_p+\rho_p b_c\le b_p+b_c\le2\beta.
\]

\smallskip
\noindent\emph{Case 4: $c$ is late and $\rho_p\ge1$.}
Here $p$ is early.  Every small job already in $X$ has ratio at least
$\rho_p$, and is therefore early.  Thus the late critical job $c$ is large and
belongs to $L_X\subseteq X^*$.

Let $\delta=\ind[p\in X^*]$.  If $p\in X^*$, remove it from $U$ and call the
result $U_0$; otherwise let $U_0=U$.  Split
\[
 U^- =\{u\in U_0:u\prec_{\mathrm J}c\},
 \qquad
 U^+ =\{u\in U_0:c\prec_{\mathrm J}u\}.
\]
Comparing the same $c$-cut in $Q$ and $X^*$ gives
\begin{equation}
 a(W)+(1-\delta)a_p>a(U^-)+b(U^+).
 \label{eq:case4-cut}
\end{equation}
The ratio order gives
\[
 b(W)\ge\rho_p a(W),
 \qquad
 b(U^-)\le\rho_p a(U^-),
\]
and $b(U^+)\le\rho_p b(U^+)$ because $\rho_p\ge1$.  Multiplying
\cref{eq:case4-cut} by $\rho_p$ and using $\rho_pa_p=b_p$ gives
\[
 b(U^-)+b(U^+)<b(W)+(1-\delta)b_p.
\]
As $b(U)=\delta b_p+b(U^-)+b(U^+)$, we obtain
$E_T=b(U)-b(W)<b_p\le\beta$.

\smallskip
\noindent\emph{Case 5: $c$ is late, $\rho_p<1$, and $b_c>\beta$.}
The job $c$ is large and lies in $L_X\subseteq X^*$.  Moreover, every small
early job precedes $c$, while every small late job follows $c$: late jobs are
ordered by nonincreasing $b_i$, and $b_c>\beta\ge b_i$ for a small late job.
Thus every small job again contributes $w_i=\min\{a_i,b_i\}$ to the fixed
$c$-cut.  The same residual-capacity knapsack argument as in Case~2 gives
$E_T\le b_p\le\beta$.

The five cases cover every type and scale of the critical job and every value of
$\rho_p$, proving the claim.
\end{proof}

\subsection{Complementary candidates in the boundary case}

\Cref{lem:tail-discrepancy} already bounds the first term of
\cref{eq:partition-makespan}.  \Cref{lem:basic-purge} handles the case in which a
large job is responsible for the second term and the boundary lies in the late
part of the ratio order.

\begin{lemma}
\label{lem:basic-purge}
Suppose that the boundary job $p$ exists, that $q\in L_Y$ is critical in
$J(Y)$, and that $\rho_p<1$.  Then
\[
 \min\{C(X,Y),C(X^-,Y^-)\}\le\frac43\OPT.
\]
\end{lemma}

\begin{proof}
Retain the sets $U=(X^*\setminus X)\cap S$ and
$W=(X\setminus X^*)\cap S$.  Since $U$ lies after the boundary and
$\rho_p<1$, all jobs in $U$ are Johnson-late.

Because $q$ is critical for the basic late set,
\begin{equation}
 a(X)+\CJ(Y)=F_q(X,Y).
 \label{eq:basic-anchor}
\end{equation}
The correct enumeration places $q$ in $Y^*$, so
\begin{equation}
 F_q(X^*,Y^*)\le a(X^*)+\CJ(Y^*)\le\OPT.
 \label{eq:opt-anchor}
\end{equation}
For a job before $q$, moving it between the two sides does not alter its
contribution to the anchored expression: it contributes $a_i$ through either
$a(X)$ or the $a$-prefix of $Y$.  Only jobs after $q$ produce a difference.
Consequently,
\begin{equation}
 F_q(X,Y)-F_q(X^*,Y^*)
 =\sum_{\substack{w\in W\\q\prec_{\mathrm J}w}}(a_w-b_w)
 -\sum_{\substack{u\in U\\q\prec_{\mathrm J}u}}(a_u-b_u).
 \label{eq:exact-anchor-difference}
\end{equation}

Let $W^{\mathrm L}=\{w\in W:a_w>b_w\}$ and
$\Delta(R)=\sum_{i\in R}(a_i-b_i)$.  Every $u\in U$ follows $q$: this is
automatic if $q$ is early, and if $q$ is late it follows from $q\in L$, the
top-six definition, and the large-before-small tie rule.  Early jobs in $W$
that follow $q$ have nonpositive difference and may be discarded from an upper
bound.  Thus
\begin{equation}
 F_q(X,Y)-F_q(X^*,Y^*)
 \le H:=\Delta(W^{\mathrm L})-\Delta(U).
 \label{eq:H}
\end{equation}
The basic candidate's anchored critical-path excess is at most $H$.

For the purge candidate, define
\begin{equation}
 Z=X^*\cap S\cap\V,
 \label{eq:Z}
\end{equation}
the small late jobs that the optimum places before maintenance.  Purging all
small late jobs eliminates the risk from $W^{\mathrm L}$ but sends $Z$ to the
late side.  Moving a small early job from $Y^*$ to $X^-$ cannot increase a cut
anchored at a common large late-side job, and moving a small late job from
$X^*$ to $Y^-$ cannot increase such a cut either.  If the new critical job is
small, it is late: because $\rho_p<1$, every small early job precedes $p$ in the
ratio order and remains in $X^-$.  Its $b$ value is at most $\beta$, and the
late-job cut bound $\CJ(Y^-)\le a(Y^-)+b_q$ applies.  The tail can increase only
by moving $Z$.  Therefore
\begin{equation}
 C(X^-,Y^-)-\OPT\le\max\{b(Z),\beta\}.
 \label{eq:purge-error}
\end{equation}

All jobs of $U$ belong to $Z$; write $Z=U\disj Z_0$.  Direct expansion gives
\[
 a(W^{\mathrm L})+a(Z)-[H+b(Z)]
 =b(W^{\mathrm L})+2\Delta(U)+\Delta(Z_0)\ge0.
\]
The sets $W^{\mathrm L}$, $Z$, and $\{q\}$ are disjoint, whence
\begin{equation}
 H+b(Z)+a_q\le a(\J)\le\OPT.
 \label{eq:resource-basic-purge}
\end{equation}

We also need an independent bound involving $a_q$.  If $q$ is late, then
$\CJ(Y)\le a(Y)+b_q$, so
\begin{equation}
 a(X)+\CJ(Y)\le a(\J)+b_q<\OPT+a_q.
 \label{eq:q-independent}
\end{equation}
If $q$ is early and the basic critical-path term already does not exceed
$\OPT$, there is nothing to prove.  Otherwise \cref{eq:H} implies $H>0$.
Put $\phi=1/\rho_p-1>0$.  Ratio ordering yields
\[
 a_w-b_w\le\phi b_w\quad(w\in W^{\mathrm L}),
 \qquad
 a_u-b_u\ge\phi b_u\quad(u\in U).
\]
Hence $H>0$ implies $b(U)-b(W)<0$.  Since an early critical job satisfies
$\CJ(Y)\le b(Y)+a_q$ and $Y^*\ne\varnothing$ in the boundary case,
\begin{align*}
 a(X)+\CJ(Y)-\OPT
 &\le s+b(Y)+a_q-[t+b(Y^*)]\\
 &=E_T+a_q-(t-s)<a_q.
\end{align*}
Thus the basic critical-path excess is below $a_q$ in both cases.

If $a_q\le\OPT/3$, this bound and
$E_T\le2\beta=2\OPT/7<\OPT/3$ show that the basic candidate is within $4/3$.
If $a_q>\OPT/3$, \cref{eq:resource-basic-purge} gives
$H+b(Z)<2\OPT/3$.  Therefore $H<\OPT/3$ or $b(Z)<\OPT/3$.  In the first case
the basic candidate is within $4/3$ by \cref{eq:H}; in the second, the purge
candidate is within $4/3$ by \cref{eq:purge-error} and
$\beta<\OPT/3$.
\end{proof}

\subsection{Three candidates without a boundary}

If no boundary job exists in the correct enumeration, then
\begin{equation}
  \CJ(L_X\cup S)\le s.
  \label{eq:no-boundary}
\end{equation}
All three early sets in
\cref{eq:no-boundary-c1,eq:no-boundary-c2,eq:no-boundary-c3} are subsets of
$L_X\cup S$ and are therefore feasible before maintenance.

\begin{lemma}
\label{lem:no-boundary}
Under the correct enumeration and \cref{eq:no-boundary}, at least one of the
three no-boundary candidates has makespan at most $4\OPT/3$.
\end{lemma}

\begin{proof}
We first normalize the optimal partition.  If a small early job
$e\in S^{\mathrm E}\cap Y^*$ is moved into $X^*$, its new early set remains a
subset of $L_X\cup S$ and hence is feasible by \cref{eq:no-boundary}.  For a
remaining anchor $q\in Y^*\setminus\{e\}$, the move changes no contribution if
$e\prec_{\mathrm J}q$ and changes $b_e$ to $a_e\le b_e$ if
$q\prec_{\mathrm J}e$.  The tail also decreases by $b_e$.  Repetition yields an
optimal partition satisfying
\begin{equation}
 S^{\mathrm E}\subseteq X^*.
 \label{eq:early-normalization}
\end{equation}

If $Y^*=\varnothing$, then $X^*=\J$.  The correct enumeration has $L_X=L$ and
Candidate~1 is exactly the classical Johnson schedule, so it is optimal.  Assume
henceforth that $Y^*\ne\varnothing$.  By
\cref{eq:correct-enumeration,eq:early-normalization}, there is a unique
decomposition
\begin{equation}
 X^*=L_X\cup S^{\mathrm E}\cup Z,
 \qquad
 Y^*=L_Y\cup W,
 \qquad
 Z\disj W=S^{\mathrm L}.
 \label{eq:normalized-opt}
\end{equation}

Candidate~1 moves every job of $W$ from $Y^*$ to its early side.  Its tail does
not increase.  Deleting a late job $i$ from a Johnson sequence decreases each
remaining anchored cut by at least $b_i$: it removes $a_i>b_i$ from an
$a$-prefix or exactly $b_i$ from a $b$-suffix.  Thus, when the late set is
nonempty,
\[
 \CJ(L_Y)\le\CJ(Y^*)-b(W),
\]
while $a(X_1)=a(X^*)+a(W)$.  The empty late-set case is directly optimal.
In either case,
\begin{equation}
 C(X_1,Y_1)-\OPT\le\Delta(W),
 \qquad
 \Delta(W)=a(W)-b(W).
 \label{eq:candidate1-error}
\end{equation}

Candidate~2 moves every job of $Z$ to the late side.  Its tail increases by at
most $b(Z)$.  A cut anchored at a common large late-side job cannot increase:
for $z\prec_{\mathrm J}q$ the contribution remains $a_z$, and for
$q\prec_{\mathrm J}z$ it changes from $a_z$ to $b_z<a_z$.  If the new critical
job is small, it is late and has $b_q\le\beta$, so
$a(X_2)+\CJ(Y_2)\le a(\J)+b_q\le\OPT+\beta$.  Hence
\begin{equation}
 C(X_2,Y_2)-\OPT\le\max\{b(Z),\beta\}.
 \label{eq:candidate2-error}
\end{equation}

For Candidate~3, decompose
\[
 W^+=W\cap D^+,
 \quad W^-=W\cap D^- ,
 \quad Z^+=Z\cap D^+,
 \quad Z^-=Z\cap D^- ,
\]
and put
\begin{equation}
 x=\Delta(W^+),
 \quad y=\Delta(W^-),
 \quad z=b(Z^+),
 \quad w=b(Z^-).
 \label{eq:xyzw}
\end{equation}
Then \cref{eq:candidate1-error,eq:candidate2-error} become
\begin{align}
 C(X_1,Y_1)-\OPT&\le x+y,
 \label{eq:c1-xyzw}\\
 C(X_2,Y_2)-\OPT&\le\max\{z+w,\beta\}.
 \label{eq:c2-xyzw}
\end{align}
Candidate~3 moves $W^+$ early and $Z^-$ late.  Its tail increases by at most
$w$.  A common large anchor incurs at most the positive discrepancy
$\Delta(W^+)=x$, whereas a new small critical job belongs to $D^-$, is late,
and has $b_q\le\beta$.  Therefore
\begin{equation}
 C(X_3,Y_3)-\OPT\le\max\{x,w,\beta\}.
 \label{eq:c3-xyzw}
\end{equation}

The ratio-$1/2$ split converts these four error quantities into first-machine
workload.  Specifically,
\[
 a(W^+)\ge2x,
 \quad a(W^-)\ge y,
 \quad a(Z^+)\ge z,
 \quad a(Z^-)>2w.
\]
The four sets are pairwise disjoint, and $a(\J)\le\OPT$, so
\begin{equation}
 2x+y+z+2w\le\OPT.
 \label{eq:xyzw-resource}
\end{equation}

Suppose all three candidates exceeded $4\OPT/3$.  Since
$\beta=\OPT/7<\OPT/3$, \cref{eq:c1-xyzw,eq:c2-xyzw,eq:c3-xyzw} would imply
\[
 x+y>\frac13\OPT,
 \qquad
 z+w>\frac13\OPT,
 \qquad
 x>\frac13\OPT\ \text{or}\ w>\frac13\OPT.
\]
If $x>\OPT/3$, then $2x+z+w>\OPT$; if $w>\OPT/3$, then
$x+y+2w>\OPT$.  Either conclusion contradicts
\cref{eq:xyzw-resource}.  At least one candidate is therefore within $4/3$.
\end{proof}

\subsection{Approximation guarantee}

\begin{theorem}
\label{thm:approximation}
Algorithm $H$ is a $10/7$-approximation for
$F2\mid\mathrm{nr\mbox{-}a}(M_2)\mid C_{\max}$.
\end{theorem}

\begin{proof}
Consider the correct enumeration.  First suppose that the boundary job exists.
Then $Y^*\ne\varnothing$, and \cref{lem:tail-discrepancy} gives
\begin{equation}
 t+b(Y)\le t+b(Y^*)+2\beta
 \le\OPT+2\beta=\frac97\OPT.
 \label{eq:unified-tail}
\end{equation}
It remains to control
\[
 M=a(X)+\CJ(Y).
\]
Let $q$ be critical in $J(Y)$.  A large critical job must lie in $L_Y$.

If $q$ is a small late job, then
$\CJ(Y)\le a(Y)+b_q$, and hence
\[
 M\le a(\J)+b_q\le\OPT+\beta=\frac87\OPT.
\]
If $q$ is a small early job, then
$\CJ(Y)\le b(Y)+a_q$ and $a_q\le b_q\le\beta$.  Using $a(X)\le s$ and
$\OPT\ge t+b(Y^*)$,
\[
 M-\OPT
 \le s+b(Y)+a_q-[t+b(Y^*)]
 =E_T+a_q-(t-s)
 \le3\beta,
\]
so $M\le10\OPT/7$.

Now let $q\in L_Y$.  If $\rho_p<1$, the stronger
\cref{lem:basic-purge} applies.  Suppose $\rho_p\ge1$.  Every job of
$W=(X\setminus X^*)\cap S$ has ratio at least one and is early.  If $q$ is
late, all such jobs precede $q$.  In the exact anchored difference
\cref{eq:exact-anchor-difference}, the positive $W$-sum is empty, and every
$u\in U$ following $q$ is late, so its subtracted difference is positive.
Thus
\[
 F_q(X,Y)\le F_q(X^*,Y^*)\le\OPT,
\]
and $M\le\OPT$.

Finally suppose that $q\in L_Y$ is early.  If $a_q\le\beta$, the same additive
calculation as for a small early critical job gives $M\le\OPT+3\beta$.  If
$a_q>\beta$, every small early job satisfies
$a_i\le b_i\le\beta<a_q$ and therefore strictly precedes $q$.  In particular,
every job in $W$ precedes $q$.  A job of $U$ that follows $q$ must be late, so
the anchored difference is again nonpositive and $M\le\OPT$.

Combining every branch with \cref{eq:unified-tail}, one generated boundary
candidate has makespan at most $10\OPT/7$.

If no boundary job exists, \cref{eq:no-boundary} holds and
\cref{lem:no-boundary} gives a candidate of makespan at most
$4\OPT/3<10\OPT/7$.  Since $H$ returns the best candidate over all
enumerations, the theorem follows.
\end{proof}

\begin{theorem}[Running time]
\label{thm:h-runtime}
Algorithm $H$ runs in polynomial time; a direct implementation takes
$O(n^2\log n)$ time and $O(n)$ auxiliary space apart from the output schedule.
\end{theorem}

\begin{proof}
There are at most $2^6=64$ large-job partitions.  The ratio and Johnson orders
are computed once.  For each enumeration, a direct scan tests at most $n$
prefixes, and a Johnson makespan can be evaluated in $O(n)$ time once the
restricted order is known.  Candidate evaluation has the same cost.  This gives
the stated conservative bound.  Incremental maintenance of cut values reduces
the repeated work, but is not needed for polynomiality.
\end{proof}

\section{A pseudopolynomial exact algorithm}
\label{sec:dp}

The canonical decomposition reduces exact optimization to choosing a subset
$X$ that fits before $s$; its complement $Y$ is scheduled after $t$.  The
partition nevertheless cannot be summarized by a single load, because the
objective contains both $b(Y)$ and $a(X)+\CJ(Y)$, while feasibility depends on
$\CJ(X)$.  We record precisely these four quantities.

Renumber the jobs $1,\ldots,n$ in the fixed global Johnson order and let
\[
 A_i=\sum_{h=1}^i a_h,
 \qquad
 A=A_n,
 \qquad
 B=\sum_{h=1}^n b_h,
 \qquad
 P=A+B.
\]
For $i\in\{0,\ldots,n\}$, define a Boolean state
\[
 R_i(x,y,p,q)
\]
to be true if the first $i$ jobs can be partitioned into subsequences $X_i$ and
$Y_i$ such that
\[
 a(X_i)=x,
 \qquad b(Y_i)=y,
 \qquad \CJ(X_i)=p,
 \qquad \CJ(Y_i)=q.
\]
Both subsequences inherit the global Johnson order.  Set
$R_0(0,0,0,0)=\mathrm{true}$ and all other layer-zero states to false.

For each true state $R_{i-1}(x,y,p,q)$, create two successor states.  Assigning
job $i$ to the early side gives
\begin{equation}
 R_i\bigl(x+a_i,\ y,\ \max\{p,x+a_i\}+b_i,\ q\bigr)=\mathrm{true}.
 \label{eq:dp-X}
\end{equation}
Assigning it to the late side gives
\begin{equation}
 R_i\bigl(x,\ y+b_i,\ p,\ \max\{q,A_i-x\}+b_i\bigr)=\mathrm{true}.
 \label{eq:dp-Y}
\end{equation}
In \cref{eq:dp-Y}, the total first-machine time of the late subsequence after
adding job $i$ is $A_i-x$.

At the last layer, inspect every true state with $p\le s$.  Its schedule value is
\begin{equation}
 \Gamma(x,y,p,q)=
 \begin{cases}
   p, & y=0,\\[1mm]
   \max\{t+y,x+q\}, & y>0.
 \end{cases}
 \label{eq:dp-objective}
\end{equation}
The positivity of the $b_i$ implies that $y=0$ exactly when the late set is
empty.

\begin{algorithm}[H]
\caption{Pseudopolynomial exact dynamic program}
\label{alg:dp}
\begin{algorithmic}[1]
\State Sort the jobs in the fixed global Johnson order and compute $A_i$.
\State Initialize $R_0(0,0,0,0)=\mathrm{true}$.
\For{$i=1$ to $n$}
  \For{each true state $R_{i-1}(x,y,p,q)$}
    \State Insert the early successor in \cref{eq:dp-X}.
    \State Insert the late successor in \cref{eq:dp-Y}.
  \EndFor
\EndFor
\State Choose a true $R_n(x,y,p,q)$ with $p\le s$ minimizing
       $\Gamma(x,y,p,q)$.
\State Recover $X,Y$ through predecessor pointers and return
       $J(X)\mid[s,t]\mid J(Y)$.
\end{algorithmic}
\end{algorithm}

\begin{theorem}[Exactness of the dynamic program]
\label{thm:dp-correct}
Algorithm~\ref{alg:dp} returns an optimal schedule.
\end{theorem}

\begin{proof}
We first prove the state invariant by induction on $i$.  It is immediate for
$i=0$.  Assume it holds at layer $i-1$.  Appending job $i$ to $X_i$ changes its
first-machine load from $x$ to $x+a_i$.  In a two-machine flow shop, the second
operation of the appended job finishes at
\[
 \max\{p,x+a_i\}+b_i,
\]
which is exactly \cref{eq:dp-X}.  If job $i$ is appended to $Y_i$, that
subsequence has first-machine load $A_i-x$ after the append, so its completion
time becomes $\max\{q,A_i-x\}+b_i$, as in \cref{eq:dp-Y}.  Thus every generated
state has the claimed interpretation.

Conversely, take any partition of the first $i$ jobs and inspect the membership
of job $i$.  Removing it yields a partition represented at layer $i-1$ by the
induction hypothesis, and the appropriate recurrence reconstructs its four
statistics.  Hence every attainable partition is represented.

By \cref{lem:canonical}, some optimal schedule has a partition $(X^*,Y^*)$ in
restricted Johnson order and satisfies $\CJ(X^*)\le s$.  The invariant places
its statistics in a terminal state admitted by Algorithm~\ref{alg:dp}; by
\cref{lem:partition-makespan}, \cref{eq:dp-objective} equals its makespan.  Every
terminal state considered by the algorithm also induces a feasible schedule of
the value in \cref{eq:dp-objective}.  Minimization over these states is therefore
exact.
\end{proof}

\begin{theorem}[Pseudopolynomial complexity]
\label{thm:dp-runtime}
For integral input data, Algorithm~\ref{alg:dp} can be implemented in
$O(nABP^2)$ time and $O(ABP^2)$ memory.  It is therefore pseudopolynomial.
\end{theorem}

\begin{proof}
The ranges are $0\le x\le A$, $0\le y\le B$, and
$0\le p,q\le P$.  There are $O(ABP^2)$ possible states per layer and two
constant-time transitions per reachable state.  Rolling arrays keep only two
layers.  Predecessor information for reconstruction may be stored for reachable
states or recovered by a second pass.  Since $A$, $B$, and $P$ are numerical
rather than binary input lengths, the bound is pseudopolynomial.
\end{proof}

\begin{remark}
The dynamic program is most naturally implemented as a sparse dictionary of
reachable tuples, with dominance filtering among states that share selected
coordinates.  Such engineering can greatly reduce practical memory use.  The
full four-dimensional formulation is retained here because its invariant and
exactness proof are immediate and do not rely on unproved dominance claims.
\end{remark}

\section{No FPTAS unless \texorpdfstring{$\mathrm{P}=\mathrm{NP}$}{P=NP}}
\label{sec:no-fptas}

A fully polynomial-time approximation scheme for a minimization problem returns,
for every $\varepsilon>0$, a feasible solution of value at most
$(1+\varepsilon)\OPT$ in time polynomial in the input length and $1/\varepsilon$.
To exclude such a scheme, it is not enough to show an additive YES/NO gap: the
gap must remain an inverse-polynomial fraction of the threshold.  We construct
exactly such a gap from \textsc{Partition}~\cite{GareyJohnson1979}.

An instance of \textsc{Partition} consists of positive integers
$x_1,\ldots,x_m$ with
\[
 \sum_{i=1}^m x_i=2B
\]
and asks whether some subset sums to $B$.  We may assume $m\ge4$ and
$x_i<B$ for all $i$.  Indeed, instances with fewer than four numbers can be
solved directly; an item equal to $B$ gives an immediate YES instance; and an
item larger than $B$ gives an immediate NO instance after the total-sum check.
These polynomial-time preprocessing cases do not affect NP-completeness of the
remaining family.

\subsection{Forcing an exact cardinality}

Append $m$ zero values and define
\[
 y_i=
 \begin{cases}
   x_i, & 1\le i\le m,\\
   0, & m+1\le i\le2m.
 \end{cases}
\]
Let
\begin{equation}
 \Lambda=(m+1)B,
 \qquad
 c_i=\Lambda+y_i,
 \qquad
 R=m\Lambda+B.
 \label{eq:cardinality-lift}
\end{equation}

\begin{lemma}
\label{lem:cardinality-lift}
There is a set $S\subseteq\{1,\ldots,2m\}$ with
$\sum_{i\in S}c_i=R$ if and only if the original
\textnormal{\textsc{Partition}} instance is
a YES instance.  Every such $S$ has $|S|=m$.
\end{lemma}

\begin{proof}
Let $k=|S|$.  If $k\le m-1$, then
\[
 \sum_{i\in S}c_i
 \le(m-1)\Lambda+2B<R.
\]
If $k\ge m+1$, then
$\sum_{i\in S}c_i\ge(m+1)\Lambda>R$.  Hence equality requires $k=m$.
For $k=m$ it is equivalent to
\[
 m\Lambda+\sum_{i\in S}y_i=m\Lambda+B,
\]
or $\sum_{i\in S}y_i=B$.  Discarding the appended zeros gives exactly a
partition subset of the original instance.  Conversely, pad any original
subset summing to $B$ with enough appended zero entries to obtain cardinality
$m$.
\end{proof}

The preprocessing assumption $0\le y_i\le B$ also gives
\begin{equation}
 c_i\le(m+2)B<\frac{R}{m-1}\le\frac{R}{3}<\frac{R}{2},
 \label{eq:ci-small}
\end{equation}
where the strict middle inequality follows from
$R/(m-1)=(m+2+3/(m-1))B$.

\subsection{Jobs, anchors, and the maintenance interval}

For every $c_i$, create an item job
\begin{equation}
 J_i=(a_i,b_i)=(R+c_i,\ 2R-c_i).
 \label{eq:item-job}
\end{equation}
Thus $a_i+b_i=3R$.  For an item subset $S$, write
$k=|S|$ and $Z=\sum_{i\in S}c_i$.  Then
\begin{equation}
 a(S)=kR+Z,
 \qquad
 b(S)=2kR-Z,
 \qquad
 a(S)+b(S)=3kR.
 \label{eq:item-loads}
\end{equation}

Let $Q=10mR$ and introduce two anchor jobs
\begin{equation}
 J_r=(1,Q),
 \qquad
 J_q=(Q+(m-2)R,Q).
 \label{eq:anchor-jobs}
\end{equation}
Finally, set
\begin{equation}
 s=1+2Q+(2m-1)R,
 \qquad
 t=s+1.
 \label{eq:reduction-hole}
\end{equation}
The unavailability interval therefore has unit length.

Let $N$ denote all $2m+2$ constructed jobs.  Since
$\sum_i c_i=2R$, the total first-machine load is
\begin{equation}
 a(N)=1+Q+3mR<t=2+2Q+(2m-1)R.
 \label{eq:all-M1-before-t}
\end{equation}
Moreover, all jobs cannot complete on $M_2$ before $s$.  At the cut anchored by
$J_r$, the full Johnson sequence has value
\[
 1+2Q+(4m-2)R>s.
\]

\begin{lemma}
\label{lem:early-set-formulation}
For the constructed instance,
\begin{equation}
 \OPT=\min_{E:\,\CJ(E)\le s}\{t+b(N\setminus E)\}.
 \label{eq:early-set-opt}
\end{equation}
\end{lemma}

\begin{proof}
In any feasible schedule, let $E$ be the jobs whose second operations finish by
$s$, and put $D=N\setminus E$.  The full set is not feasible before $s$, so
$D\ne\varnothing$.  Deleting $D$ leaves a two-machine schedule for $E$ that
finishes by $s$; Johnson optimality gives $\CJ(E)\le s$.  Every job in $D$ must
receive its entire second operation after $t$, because any processing interrupted
at $s$ is lost.  Hence $C_{\max}\ge t+b(D)$.

Conversely, take any $E$ with $\CJ(E)\le s$.  Schedule $J(E)$ before the
interval and continue processing the first operations of $D$ on $M_1$.  By
\cref{eq:all-M1-before-t}, all first operations are complete by $t$.  Starting
at $t$, process all second operations of $D$ consecutively.  This produces a
feasible permutation schedule of value $t+b(D)$.  Minimization proves
\cref{eq:early-set-opt}.
\end{proof}

Equation~\eqref{eq:early-set-opt} says that an optimal early set maximizes
$b(E)$ subject to Johnson feasibility.

\subsection{The two feasibility inequalities}

By \cref{eq:ci-small}, each item job is Johnson-early because
$R+c_i<2R-c_i$.  In every set containing both anchors, the restricted Johnson
order is
\[
  J_r,\quad \text{selected item jobs},\quad J_q.
\]
Indeed, $J_r$ is early and comes first, while $J_q$ is late and comes last.

For a job $j$ in a Johnson sequence, let
\[
 H_j=\sum_{i\preceq_{\mathrm J}j}a_i
     +\sum_{j\preceq_{\mathrm J}i}b_i.
\]
Suppose two selected item jobs $i,j$ are consecutive with $i$ before $j$.
Then
\[
 H_j-H_i=a_j-b_i=c_i+c_j-R<0
\]
by \cref{eq:ci-small}.  From $J_r$ to the first item,
$H_i-H_r=a_i-b_r=R+c_i-Q<0$.  Hence all item cuts are dominated by the two
anchor cuts.

\begin{lemma}
\label{lem:anchor-feasibility}
For an item subset $S$,
\begin{equation}
 \CJ(\{J_r,J_q\}\cup S)\le s
 \quad\Longleftrightarrow\quad
 \begin{cases}
   a(S)\le(m+1)R,\\
   b(S)\le(2m-1)R.
 \end{cases}
 \label{eq:two-anchor-constraints}
\end{equation}
\end{lemma}

\begin{proof}
The preceding cut comparison gives
$\CJ(\{J_r,J_q\}\cup S)=\max\{H_r,H_q\}$.  Direct calculation yields
\[
 H_r=1+2Q+b(S),
 \qquad
 H_q=1+2Q+(m-2)R+a(S).
\]
Comparing both expressions with $s$ from \cref{eq:reduction-hole} gives exactly
the two inequalities in \cref{eq:two-anchor-constraints}.
\end{proof}

Adding the two constraints and using \cref{eq:item-loads} shows that every
feasible anchored item set has
\[
 3kR=a(S)+b(S)\le3mR,
 \qquad\text{hence }k\le m.
\]
If $k=m$, the two inequalities further imply
\[
 mR+Z\le(m+1)R
 \quad\text{and}\quad
 2mR-Z\le(2m-1)R,
\]
so $Z=R$.  The converse is immediate.  Therefore
\begin{equation}
 \{J_r,J_q\}\cup S\text{ is early-feasible with }|S|=m
 \quad\Longleftrightarrow\quad
 \sum_{i\in S}c_i=R.
 \label{eq:encoded-partition}
\end{equation}

Both anchors must occur in every optimal early set.  The pair
$\{J_r,J_q\}$ is feasible and has early $b$-load $2Q$.  An early set omitting at
least one anchor has
\[
 b(E)\le Q+(4m-2)R<2Q,
\]
because $(4m-2)R<10mR=Q$.  Such a set cannot maximize $b(E)$.

\subsection{The inverse-polynomial gap}

Define
\begin{equation}
 K=t+2mR.
 \label{eq:gap-threshold}
\end{equation}

\begin{lemma}
\label{lem:gap}
If the \textnormal{\textsc{Partition}} instance is a YES instance, then
$\OPT=K-R$.  If it is a NO instance, then $\OPT>K$.
\end{lemma}

\begin{proof}
In a YES instance, \cref{lem:cardinality-lift} gives an item subset $S$ with
$|S|=m$ and $Z=R$.  It is feasible with both anchors by
\cref{eq:encoded-partition}, and its item $b$-load is
$b(S)=(2m-1)R$.  No anchored feasible set with fewer than $m$ items has as much
item $b$-load, since for $k\le m-1$,
$b(S)=2kR-Z<2(m-1)R$.  Every optimal early set contains both anchors, so the
maximum early load uses such an $m$-item set.  The total item $b$-load is
$(4m-2)R$; consequently the late load is $(2m-1)R$ and
\[
 \OPT=t+(2m-1)R=K-R.
\]

Now suppose the instance is NO.  If an early set omits an anchor, its late side
contains at least $Q>2mR$ units of second-machine work.  If it contains both
anchors, \cref{eq:encoded-partition} rules out $k=m$, and therefore $k\le m-1$.
Its early item load is strictly less than $2(m-1)R$, so its late item load is
strictly greater than
\[
 (4m-2)R-2(m-1)R=2mR.
\]
In either case \cref{lem:early-set-formulation} gives $\OPT>t+2mR=K$.
\end{proof}

The gap is an inverse-polynomial fraction of the threshold.  From
\cref{eq:reduction-hole,eq:gap-threshold},
\begin{equation}
 t=2+(22m-1)R,
 \qquad
 K=2+(24m-1)R,
 \qquad
 \frac{K}{R}=O(m).
 \label{eq:relative-gap}
\end{equation}
All constructed numbers have binary encoding length
$O(\log m+\log B)$ beyond the source encoding: $R=(m^2+m+1)B$ and $Q=10mR$.
Thus the reduction is polynomial.

\begin{theorem}[FPTAS lower bound]
\label{thm:no-fptas}
Unless $\mathrm P=\mathrm{NP}$, the problem
$F2\mid\mathrm{nr\mbox{-}a}(M_2)\mid C_{\max}$ has no FPTAS.  This holds even
when $t-s=1$.
\end{theorem}

\begin{proof}
Assume an FPTAS exists and run it on the constructed instance with
\[
 \varepsilon=\frac{1}{50m}.
\]
The running time is polynomial in the source input length because
$1/\varepsilon=50m$.  In a YES instance,
\[
 C^{\mathrm{FP}}
 \le(1+\varepsilon)(K-R)<K.
\]
For the strict inequality, observe that
\[
 \varepsilon(K-R)
 =\frac{2+(24m-2)R}{50m}<R.
\]
In a NO instance, every feasible schedule has
$C^{\mathrm{FP}}\ge\OPT>K$.  Comparing the returned value with $K$ therefore
decides \textsc{Partition} in polynomial time, a contradiction unless
$\mathrm P=\mathrm{NP}$.  Unit interval length follows directly from
\cref{eq:reduction-hole}.
\end{proof}

\section{Discussion}
\label{sec:discussion}

The three results expose two different forms of structure.  Johnson's rule
continues to solve every fixed subset exactly, and this local order is strong
enough to support both the approximation algorithm and the exact dynamic
program.  The global decision is harder: assigning a job to the early side can
reduce the late $b$-tail while increasing an anchored critical path by
$a_i-b_i$.  Algorithm $H$ makes these effects comparable by isolating six
large $b$-jobs and forcing every remaining second-machine contribution below
$\OPT/7$.

The dynamic program handles the same conflict without approximation by keeping
both Johnson completion times.  Its numerical state space is finite and
pseudopolynomial, but a standard one-dimensional knapsack trimming argument
cannot preserve both critical-path coordinates and the exact boundary
$\CJ(X)\le s$.  The reduction makes the obstruction explicit.  The two anchors
create two simultaneously tight constraints, one on $a(S)$ and one on $b(S)$;
rounding either coordinate can change whether $m$ item jobs fit before the fixed
interval.  Because crossing this boundary moves an entire $b$-operation into the
late tail, the resulting objective gap is already $\Theta(1/m)$ in relative
terms.

This also clarifies the asymmetry between the two possible locations of the
non-availability interval.  For the non-resumable first-machine model, an FPTAS
is known~\cite{LuEtAl2026}.  For the second-machine model, the anchors in
\cref{sec:no-fptas} couple the classical flow-shop critical path to the
post-maintenance load, excluding an FPTAS.  Exchanging $a_i$ and $b_i$ and
reversing a Johnson sequence does not preserve the fixed-time interaction, so
the classical reversal symmetry of the unconstrained two-machine flow shop does
not transfer to this setting.

Several questions remain open.  The $10/7$ ratio is an upper bound rather than a
matching threshold.  Improving it would require either a sharper use of the six
large-job enumeration or a new family of complementary partitions.  On the
exact side, the four-dimensional dynamic program favors clarity over state-space
minimality; safe dominance rules or a more compact multiobjective
representation could improve its practical complexity without implying an
FPTAS.  Finally, additional fixed intervals or non-resumable interruptions on
both machines may produce stronger inapproximability, but they also require new
canonical-decomposition arguments.

\section{Conclusion}
\label{sec:conclusion}

We developed a unified theory for
$F2\mid\mathrm{nr\mbox{-}a}(M_2)\mid C_{\max}$.  Algorithm $H$ achieves a
$10/7$ approximation by separating the placement of six large jobs from a ratio
ordered small-job partition and then sequencing every candidate by Johnson's
rule.  A four-dimensional reachability dynamic program computes an exact optimum
in pseudopolynomial time.  Finally, a two-anchor reduction from
\textsc{Partition} produces an inverse-polynomial YES/NO gap and rules out an
FPTAS unless $\mathrm P=\mathrm{NP}$, even for a unit-length interruption.  The
combination of these results pinpoints the second machine as a genuine
approximability boundary in the non-resumable two-machine flow shop.

\backmatter

\bibliography{algorithmica-article}

\end{document}